\documentclass[11pt,a4paper]{article}
\usepackage{amsmath,amssymb,latexsym,amsbsy,amsthm}
\usepackage{epsfig,graphicx}

\newtheorem*{theorem}{Theorem}
\newtheorem*{lemma}{Statement}

\begin{document}

\title{\bf Haag's theorem  in $S\,O\,(1, k)$
invariant quantum field theory\footnote{Talk
given at the XVIth International Seminar on High Energy Physics QUARKS'2010,
Kolomna, Russia, 6-12 June, 2010 and the XIXth International Workshop on High Energy Physics and Quantum Field Theory, Golitsyno, Moscow, Russia, 8-15 September 2010.}}

\author{K.~V.~ Antipin$^{a}$, Yu.~S.~Vernov$^{b}$, M.~N.~Mnatsakanova$^{c}$
\\
$^a$ \small{\em Faculty of Physics,  Moscow State University, Moscow, Russia.} \\
$^b$ \small{\em Institute for Nuclear Research, Russian Academy of Sciences,} \\
\small{\em 60-the October Anniversary prospect 7a, Moscow, 117312, Russia. }\\
$^c$ \small{\em Institute of Nuclear Physics, Moscow State University,} \\
\small{\em  119992, Vorobyevy Gory, Moscow, Russia. }}
\date{}
\maketitle{}

\begin{abstract}
Generalized Haag's theorem has been proven in $S\,O\,(1,
k)$ invariant quantum field theory. Apart from the above mentioned
$k + 1$ variables there can be arbitrary number of additional
coordinates including noncommutative ones in the theory. New
consequences of generalized Haag's theorem are obtained.

It has been proven that the equality of four-point Wightman
functions in two theories leads to the equality of elastic
scattering amplitudes and thus the total cross-sections in these
theories.

In space-space noncommutative quantum field theory  in
four-dimensional case it has been proved that if in one of the
theories under consideration $S$-matrix is equal to unity, then in
another theory $S$-matrix is unity as well.
\end{abstract}


\section{Introduction}
In this report we consider  Haag's theorem -- one of the most
important results of axiomatic approach in quantum field
theory~(see \cite{SW, BLT}). In the usual Hamiltonian formalism it
is assumed that asymptotic fields are
related with interacting fields by unitary transformation. \\
Haag's theorem shows that in accordance with Lorentz invariance of
the theory the interacting fields are also trivial which means
that corresponding S-matrix is equal to unity. So the usual
interaction representation can not exist in the Lorentz invariant
theory. In four dimensional case in accordance with the
generalized Haag's theorem four first Wightman functions coincide
in two related by the unitary transformation  theories.

Let us recall that n-point Wightman functions $W\,(x_1,  \ldots,
x_n) $ are $\langle \, \Psi_0, \varphi (x_1) \ldots \varphi (x_n)
\, \Psi_0 \, \rangle$, where $ \Psi_0$ is a vacuum vector.
Actually in accordance with translation invariance Wightman
functions are functions of variables $\xi_i = x_i - x_{i +1}$, $i
= 1, \ldots, n-1$. At first Haag's theorem is proved in $S\,O \,
(1,3)$ invariant theory in four dimensional case.

In last years noncommutative generalization of QFT - NC QFT -
attracts interest of physicists as in some cases NC QFT is the low
energy limit of string theory~(see \cite{SeWi,DN,Sz}). Haag's
theorem in NC QFT was considered in \cite{CPT, CMTV}.

NC QFT is defined by the Heisenberg-like commutation relations
between coordinates
\begin{equation} \label{cr}
[ {x}^{\mu}, {x}^{\nu}] =i \,\theta^{\mu \nu},
\end{equation}
where $\theta^{\mu\nu}$ is a constant antisymmetric matrix. \\
It is very important that NC QFT can be formulated in commutative
space if the usual product between operators (precisely between
corresponding test functions) is substituted by the star
(Moyal-type) product.

Now the theories in spaces of arbitrary dimensions are widely
considered. Thus it is interesting to consider Haag's theorem in
the general case of $k+1$ commutative variables (time and $k$
spatial coordinates) and arbitrary number $m$ of noncommutative
coordinates. This extension of Haag's theorem is a goal of our
work.

\section{Generalized Haag's theorem in four dimensional
case}
In axiomatic QFT  fields $\varphi\,(f)$  smeared on all four
coordinates are unbounded operators in the state vectors space. In
the derivation of Haag's theorem it is necessary to assume that
fields smeared on the spatial coordinates are well defined
operators.

Let us recall generalized Haag's theorem in four dimensional case.

\begin{theorem}
Let $\varphi_{1}(f,t)$  and $\varphi_{2}(f,t)$  belong to two sets
of irreducible operators in Hilbert spaces $H_{1}$ and $H_{2}$
correspondingly. We assume that both theories are Poincare
invariant, that is
\begin{equation}\label{h1}
U_j(a,\Lambda)\,\varphi_j(x)U_j^{-1}(a,\Lambda)=\varphi_j(\Lambda
x+a),
\end{equation}
\begin{equation}\label{h2}
U_j(a,\Lambda)\Psi_{0j}=\Psi_{0j},\qquad j=1,2.
\end{equation}
Let us suppose also that there exists  unitary transformation $V$
relating fields and vacuum states as well in two theories at any
$t$:
\begin{gather}
\varphi_2(f,t)=V\varphi_1(f,t)V^{-1}\label{h3},\\
c\Psi_{02}=V\Psi_{01}\label{h4},
\end{gather}
where $c$ is a complex number with module one.

Let as suppose that  vacuum states in two theories are invariant
under the same unitary transformation.

If in two theories there are not states with negative energies
then four first Wightman functions coincide in two theories.
\end{theorem}

It is important to note that, given rather common assumptions,
condition~(\ref{h4}) is a direct consequence of eq.~(\ref{h3})
(see {\bf Statement} at the end of the next section).

Let us give the idea of the proof.

The invariance of the vacuum states implies that in accordance
with conditions (\ref{h1}) and (\ref{h2}) Wightman functions
coincide at equal times.
\begin{equation}
\left(\Psi_{01},\varphi_1(t,x_1),\ldots,\varphi_1(t,x_n)\Psi_{01}\right)=
\left(\Psi_{02},\varphi_2(t,x_1),\ldots,\varphi_2(t,x_n)\Psi_{02}\right)
\end{equation}
Let us recall that  Wightman functions are analytical functions in
tubes ${T}_{n}^{-}: \; \nu \in {T}_{n}^{-}$ if $\nu = \xi - i\,
\eta, \; \xi$ is arbitrary, $\eta \in V^{+}$, which means that
$\eta_{0}^{2} - {\vec{\eta}}^{2} > 0.$

It can be shown that the equality  of Wightman functions at equal
times and their analyticity  lead to equality of four first
Wightman functions in two theories related by unitary
transformation at all points.

Let us point out that noncommutative coordinates belong to the
boundary of analyticity of Wightman functions. As in the
derivation of Haag's theorem only transformations of coordinates
which belong to the domain of analyticity are essential, we omit
the dependence of vectors on additional variables.

\section{Extension of generalized Haag's theorem}
Let us obtain the extension of Haag's theorem on the $S\,O(1, k)$
invariant theory.

As at $n > k$ vectors $\xi_{i} = (0, \vec{\xi}_{i})$ are linearly
dependent, then vectors related to them with Lorentz
transformation are linearly dependent too and thus can not form
the open set. Thus they can not determine Wightman functions.

Let us show that if $n \leq k$ then Wightman functions coincide
in two theories under consideration.

Indeed, as the vectors $\vec{\xi}_{i}$ are arbitrary we can choose
vector $\xi_{2} = (0, \vec{\xi}_{2})$ in such a way that it be
orthogonal  to $\xi_{1}$. Continuing this procedure we obtain the
set of vectors $\xi_{i}$ in such a way that they all be orthogonal
one to another.

As $\xi_{i} \bot \xi_{j}$ if $i \neq j$, then also
$\alpha\,\xi_{i} \bot \beta\,\xi_{j}, \; \alpha, \beta \in
\mathbb{R}$ are arbitrary. If $\tilde{\xi_{i}} = L\,\xi_{i}$,
where $L$ is a real Lorentz transformation, then $\tilde{\xi_{i}}
\bot \tilde{\xi_{j}}, \; i \neq j$ and also
$\alpha\,\tilde{\xi_{i}} \bot \beta\,\tilde{\xi_{j}}$. So these
points form the open subset and thus fully determine Wightman
functions owing to their analiticity.

As in two theories related by unitary transformation at equal
times first $k+1$  Wightman functions coincide on the open set,
then these functions coincide in all points.
\begin{lemma}
Condition (\ref{h4}) holds if
vacuum vectors $\Psi_0^i$ are translation invariant normalized
states with respect to shifts $U_i\,(a)$ along one of the
coordinates, satisfying  $S\,O\,(1, k)$ symmetry.
\end{lemma}
\begin{proof}
Indeed, it is easy to note that operator $ U_1^{- 1}\,(a)\,V^{-
1}\,U_2\,(a)\,V$ commutes with operators $\tilde {\varphi}_1 \,(t,
\vec{x})$ and, thanks to irreducibility of the operator set, is
proportional to identity operator. In limit $a = 0$ we obtain that
\begin{equation}\label{34}
U_1^{- 1}\,(a)\,V^{- 1}\,U_2\,(a)\,V = \mathbb{I}.
\end{equation}
From eq.~(\ref{34}) it follows that if
\begin{equation}\label{35}
U_1\,(a)\,\Psi_0^1 = \Psi_0^1,
\end{equation}
then
\begin{equation}\label{36}
U_2\,(a)\,V\,\Psi_0^1 = V\,\Psi_0^1,
\end{equation}
that is condition~(\ref{h4}) is satisfied.
\end{proof}

\section{Consequences of Haag's theorem}
Let us obtain consequences of Haag's theorem in $S \, O (1,
k)$ invariant theory $k \geq 3$.  In accordance with reduction
formula
\begin{multline}
\label{red} \left\langle p_1'\ldots p_l'{}^{\text{\emph{out}}}\mid
p_1 \ldots p_m^{\text{\emph{in}}}\right\rangle =\\
i^{l+m}\int dy_1\ldots dy_l\, dx_1\ldots dx_m\,
f_{p_1'}^*(y_1)\times\ldots\\
\times f_{p_l'}^\star(y_l)\vec K_{y_1} \ldots \vec
K_{y_l}\Bigl(\Psi_0,T\bigl(\varphi(y_1)
\ldots \varphi(x_m)\bigr)\Psi_0\Bigr)\times\\
\times \vec K_{x_1} \ldots \vec K_{x_m}f_{p_1}(x_1)\ldots
f_{p_m}(x_m),
\end{multline}
where $K_x=\left(\partial_{\mu}\partial^{\mu}\right)_x+m^2$ is
Klein-Gordon operator and $f_p(x)=\frac{e^{-ipx}}{(2\pi)^{3/2}}$
is a corresponding wave function.

Let us consider $2 \Rightarrow k - 1$ processes. We see that
amplitudes of these processes coincide in two theories.

From the equality
$$
W_1 \, (x_1, \ldots, x_4) = W_2 \, (x_1, \ldots, x_4)
$$
it follows that
\begin{equation} \label{42}
< p_3, p_4 | p_1, p_2 >_{1} = < p_3, p_4 | p_1, p_2 >_{2}
\end{equation}
for any $p_{i}$. Having applied this equality for the forward
elastic scattering amplitudes, we obtain that, according to the
optical theorem, the total cross-sections for the fields
${\varphi}_1 \, (x)$ and ${\varphi}_2 \, (x)$ coincide. If now the
$S$-matrix for the field ${\varphi}_1 \, (x)$ is unity, then it is
also unity for field ${\varphi}_2 \, (x)$.

Let us consider Haag's theorem in the $S \, O (1, 1)$ invariant
field theory. In accordance with previous result equality of only
two-point Wightman functions takes place. Let us prove that if one
of considered theories is trivial, that is the corresponding
S-matrix is equal to unity, then another is trivial too.

Let us point out that in the $S \, O (1, 1)$ invariant theory it
is sufficient that the spectral condition, which implies non
existence of tachyons, is valid only in respect to commutative
coordinates. Also it is sufficient that  translation invariance is
valid only in respect to the commutative coordinates. The equality
of two-point Wightman functions in the two
theories leads to the following conclusion: \\
if local commutativity condition in respect to commutative
coordinates is fulfilled and the current in one of the theories is
equal to zero, then another current is zero as well. \\ Indeed as
$W_{1} \, (x_1, x_2) = W_{2} \, (x_1, x_2)$, then also
\begin{equation} \label{43}
< \Psi_0^1,  j^1_{\bar{f}} \,  j^1_{f} \,  \, \Psi_0^1 > = <
\Psi_0^2, j^2_{\bar{f}} \,  j^2_{f}  \, \Psi_0^2> = 0,
\end{equation}
where
$$
j_f^{i} = (\square + m^{2}) \, \varphi_f^{i}.
$$
If, for example, $j_f^1 = 0$, then in the positive metric case
\begin{equation} \label{zero_cur}
j_f^2  \, \Psi_0^2 = 0.
\end{equation}
From the last formula and local commutativity condition it follows
that \cite{BLT}
\begin{equation}\nonumber
j_f^2  \equiv  0.
\end{equation}
Our statement is proved.

\section{Conclusions}
Generalized Haag's theorem has been proven in $S\,O\,(1, k)$
invariant quantum field theory. Apart from the above mentioned $k
+ 1$ variables there can be arbitrary number of additional
coordinates including noncommutative ones in the theory.

In $S\,O\,(1, k)$ invariant theory new consequences of generalized
Haag's theorem are obtained. It has been proven that the equality
of four-point Wightman functions in two theories leads to the
equality of elastic scattering amplitudes and thus to the equality
of the total cross-sections in these theories. Also it has been
shown that at $ k > 3$ the equality of $(k+1)$~- point Wightman
functions in two theories leads to the equality of scattering
amplitudes of some inelastic processes.

In $S\,O\,(1, 1)$  invariant theory it has been proved that if in
one of the theories under consideration $S$-matrix is equal to
unity, then in another theory $S$-matrix is unity as well.

\end{document}